\documentclass{llncs}
\usepackage[utf8]{inputenc}
\usepackage{amssymb}
\usepackage{amsmath}
\usepackage{amstext} 
\usepackage{array}   
\usepackage{stmaryrd} 
\usepackage[all,cmtip]{xy}
\usepackage{color}
\usepackage{charter}

\renewenvironment{proof} {\textsc{Proof}\quad} {\hfill $\Box$\\}

\newcommand{\Khm}[1]{\ensuremath{\mathcal{K}hm{(#1)}}}
\newcommand{\Khmm}{{\ensuremath{\mathcal{K}hm}}}
\newcommand{\Kh}[1]{\ensuremath{\mathcal{K}h{(#1)}}}

\newcommand{\U}{\mathcal{U}}
\newcommand{\lr}[1]{\langle #1 \rangle}
\newcommand{\lra}{\leftrightarrow}

\newcommand{\rel}[1]{\xrightarrow{#1}}

\newcommand{\calS}{\mathcal{S}}
\newcommand{\calM}{\mathcal{M}}
\newcommand{\M}{\calM}
\newcommand{\calR}{\mathcal{R}}
\newcommand{\calV}{\mathcal{V}}
\newcommand{\V}{\calV}
\newcommand{\Lef}[1]{{\ensuremath{\mathtt{L}(#1)}}}
\newcommand{\Rig}[1]{{\ensuremath{\mathtt{R}(#1)}}}

\newcommand{\MCS}{{\ensuremath{\Phi_\Gamma}}}

\newcolumntype{L}{>{$}l<{$}} 

\newcommand{\AxRepKhm}{\ensuremath{\mathtt{UKhm}}}
\newcommand{\ThKhmL}{\texttt{ULKhm}}
\newcommand{\ThKhmM}{\texttt{UMKhm}}
\newcommand{\ThKhmR}{\texttt{URKhm}}
\newcommand{\EMPKhm}{\texttt{EMPKhm}}
\newcommand{\COMP}{\texttt{COMPKhm}}

\newcommand{\KHDisjunction}{\texttt{ONEKhm}}
\newcommand{\Univers}{\texttt{UNIV}}
\newcommand{\REU}{\texttt{REU}}

\newcommand{\RE}{\texttt{RE}}

\newcommand{\BP}{\ensuremath{\mathbf{P}}}
\newcommand{\LKhm}{\mathbf{L_{Khm}}}
\newcommand{\Act}{\ensuremath{\mathbf{\Sigma}}}

\newcommand{\R}{\mathcal{R}}
\renewcommand{\S}{\mathcal{S}}
\newcommand{\TAUT}{\ensuremath{\mathtt{TAUT}}}
\newcommand{\DISTU}{\ensuremath{\mathtt{DISTU}}}
\newcommand{\COMPKh}{\ensuremath{\mathtt{COMPKh}}}
\newcommand{\WSKh}{\ensuremath{\mathtt{UKh}}}
\newcommand{\MP}{\ensuremath{\mathtt{MP}}}
\newcommand{\SUB}{\ensuremath{\mathtt{SUB}}}
\newcommand{\NECU}{\ensuremath{\mathtt{NECU}}}

\newcommand{\AxTransKU}{\ensuremath{\mathtt{4KhmU}}}
\newcommand{\AxEucKU}{\ensuremath{\mathtt{5KhmU}}}

\newcommand{\AxTrU}{\ensuremath{\mathtt{TU}}}
\newcommand{\AxTransU}{\ensuremath{\mathtt{4U}}}
\newcommand{\AxEucU}{\ensuremath{\mathtt{5U}}}
\newcommand{\EMP}{\ensuremath{\mathtt{EMPKh}}}
\newcommand{\SKhm}{\mathbb{SKHM}}
\newcommand{\SKh}{\mathbb{SKH}}

\renewcommand{\phi}{\varphi}
\title{Achieving while maintaining:}
\subtitle{A logic of knowing how with intermediate constraints}
\author{Yanjun Li\inst{1} \and Yanjing Wang\inst{2} }
\institute{Faculty of Philosophy, University of Groningen, The Netherlands \and Department of Philosophy, Peking University, China}
\date{}
\begin{document}
	\maketitle
	
\begin{abstract}
In this paper, we propose a ternary knowing how operator to express that the agent knows how to achieve $\phi$ given $\psi$ while maintaining $\chi$ in-between. It generalizes the logic of goal-directed knowing how proposed by Wang in \cite{Wang15lori}. We give a sound and complete axiomatization of this logic. 
\end{abstract}    
    
\section{Introduction}
Standard epistemic logic proposed by von Wright and Hintikka studies propositional knowledge expressed by ``knowing that $\phi$'' \cite{Wright51,Hintikka:kab}. However, there are very natural knowledge expressions beyond ``knowing that'', such as ``knowing what your password is'', ``knowing why he came late'', ``knowing how to go to Beijing'', and so on. In recent years, there have been attempts to capture the logic of such different kinds of knowledge expressions by taking the ``knowing X'' as a single modality \cite{WangF13,WangF14,FWvD14,FWvD15,GW16,Wang15lori}.\footnote{See \cite{Wang16} for a survey.} 

In particular, Wang proposed a logical language of goal-directed knowing how \cite{Wang15lori}, which includes formulas $\Kh{\psi, \phi}$ to express that the agent knows how to achieve $\phi$ given the precondition $\psi$.\footnote{See \cite{Gochet13,KandA15,Wang15lori} for detailed discussions on related work in AI and Philosophy.} The models are labeled transition systems which represent the agent's abilities, inspired by \cite{Wang15}. Borrowing the idea from conformant planning in AI (cf. e.g., \cite{SW98,YLW15}), $\Kh{\psi, \phi}$ holds globally in a labeled transition system, if there is an uniform plan such that from all the $\psi$-states this plan can always be successfully executed to reach some $\phi$-states. As an example, in the following model $\Kh{p, q}$ holds, since there is a plan $ru$ which can always work to reach a $q$-state from any $p$-state.   
$$
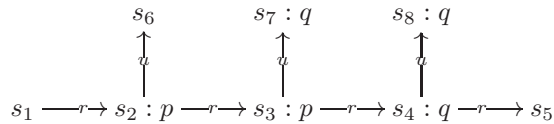
\xymatrix{
&s_6&{{s_7:q}}&{{s_8: q}} &\\
s_1\ar[r]|r& s_2:p\ar[r]|r\ar[u]|u& s_3:p\ar[r]|r\ar[u]|u&{s_4:q}\ar[r]|r\ar[u]|u&s_5
}$$
In \cite{Wang15lori}, a sound and complete proof system is given, featuring a crucial axiom capturing the compositionality of plans: 
$$\COMPKh  \qquad \Kh{p, r}\land\Kh{r, q}\to\Kh{p, q}$$

However, as observed in \cite{LauWang}, constraints on how we achieve the goal often matter. For example, the ways for me to go to New York are constrained by the money I have; we want to know how to win the game by playing fairly; people want to know how to be rich without breaking the law. Generally speaking, actions have costs, both financially and morally, we need to stay within our ``budget'' in reaching our goals. Clearly such intermediate constraints cannot be expressed by $\Kh{\psi, \phi}$ since it only cares about the starting and ending states. This motivates us to introduce a ternary modality $\Kh{\psi,\chi,\phi}$ where $\chi$ constrains the intermediate states.\footnote{This ternary modality is first proposed and discussed briefly in the full version of \cite{Wang15lori}, which is under submission to a journal.}

In the rest of the paper, we first introduce the language, semantics, and a proof system of our logic in Section \ref{Sec.Logic}. In Section \ref{Sec.Proof} we give the highly non-trivial completeness proof of our system, which is much more complicated than the one for the standard knowing how logic. In the last section we conclude with future directions. 

	
	\section{The Logic}\label{Sec.Logic}
	\begin{definition}[Language]
		Given a set of proposition letters $\BP$, the language $\LKhm$ is  defined as follows:  
		$$\phi:=p\mid \neg\phi\mid (\phi\land\phi)\mid \Khm{\phi,\phi,\phi} $$
		where $p\in\BP$. $\Khm{\psi,\chi,\phi}$ expresses that the agent knows how to guarantee $\phi$ given $\psi$ while maintaining $\chi$ in-between (excluding the start and the end). Note that $\Khm{\psi\land\chi,\chi,\phi\land\chi}$ expresses knowing how with inclusive intermediate constraints. We use the standard abbreviations $\bot, \phi\lor\psi$ and $\phi\to\psi $, and  define $\U\varphi$ as $ \Khm{\neg\phi,\top,\bot}$. $\U$ is intended to be an universal modality, and it will become more clear after defining the semantics. Note that the binary know-how operator in \cite{Wang15} can be defined as $\Kh{\psi,\phi}:= \Khm{\psi, \top, \phi}$.
\end{definition}
	
	\begin{definition}[Model]
		Given a countable set of proposition letters $\BP$ and a countable non-empty set of action symbols $\Act.$
		A model (also called an ability map) is essentially a labelled transition system $(\S, \R,\V)$ where:
		\begin{itemize}
			\item $\S$ is a non-empty set of states;
			\item $\R: \Act\to 2^{\S\times \S}$ is a collection of transitions labelled by actions in $\Act$;
			\item $\V: S\to 2^\BP$ is a valuation function.
		\end{itemize}
		We write $s\rel{a}t$ if $(s, t)\in \R(a).$ For a sequence $\sigma=a_1\dots a_n\in\Act^*$, we write $s\rel{\sigma}t$ if there exist $ s_2\dots s_{n}$ such that $s\rel{a_1}s_2\rel{a_2}\cdots \rel{a_{n-1}}s_n\rel{a_n}t$. Note that $\sigma$ can be the empty sequence $\epsilon$ (when $n=0$), and we set $s\rel{\epsilon}s$ for any $s$. Let $\sigma_k$ be the initial segment of $\sigma$ up to $a_k$ for $k\leq |\sigma|$. In particular let $\sigma_0=\epsilon$. We say $\sigma=a_1\cdots a_n$ is strongly executable at $s'$ if for each $0\leq k <n$: 
$s'\rel{\sigma_k}t$ implies that $t$ has at least one $a_{k+1}$-successor.
	\end{definition}
    Intuitively, $\sigma$ is \textit{strongly executable} at $s$ if you can always successfully finish the whole $\sigma$ after executing any initial segment of $\sigma$ from $s$. For example, $ab$ is not strongly executable at $s_1$ in the model below, though it is executable at $s_1$.
$$
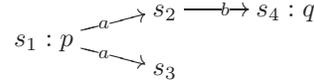
\xymatrix@R-25pt{
&{s_2}\ar[r]|b&{s_4: q}\\
{s_1:p}\ar[ur]|a\ar[dr]|a\\
&{s_3}
}
$$	
	\begin{definition}[Semantics] 
		Suppose $s$ is a state in a model $\M=(\S,\R,\V)$. Then we inductively define the notion of a formula $\phi$ being satisfied (or true) in $\M$ at state $s$ as follows: 
		\begin{center}
			\begin{tabular}{|L L L|}
				\hline
				\M,s\vDash \top &  & always\\
				\M,s\vDash p &\iff& s\in \V(p).\\
				\M,s\vDash \neg\phi &\iff& \M,s\nvDash \phi.\\
				\M,s\vDash \phi\wedge\psi &\iff& \M,s\vDash\phi \text{ and } \M,s\vDash \psi.\\
				\M,s\vDash \Khm{\psi,\chi,\phi} &\iff&\text{there exists }\sigma\in\Act^*\text{ such that for each $s'$ with}\\
                &&\M,s'\vDash\psi\text{ we have $\sigma$ is strongly $\chi$-executable}\\ 
                &&\text{at $s'$ and $\M,t\vDash \phi$ for all $t$ with $s'\rel{\sigma}t$.} \\
				\hline
			\end{tabular}
		\end{center}
		where we say $\sigma=a_1\cdots a_n$ is strongly $\chi$-executable at $s'$ if:	\begin{itemize}
        \item $\sigma$ is strongly executable at $ s'$, and
        \item $s'\rel{\sigma_k}t$ implies $\M,t\vDash \chi$ for all $0<k<n$.
		\end{itemize}
\end{definition}

It is obvious that $\epsilon$ is strongly $\chi$-executable at each state $s$ for each formula $\chi$. Note that $\Khm{\psi,\bot,\phi}$ expresses that there is $\sigma\in\Act\cup\{\epsilon\}$ such that the agent knows doing $\sigma$ on $\psi$-states can guarantee $\phi$, namely the witness plan $\sigma$ is at most one-step. As an example, $\Kh{p,\bot,o}$ and $\Kh{p,o,q}$ hold in the following model for the witness plans $a$ and $ab$ respectively. Note that the truth value of $\Kh{\psi,\chi,\phi}$ does not depend on the designated state.
$$
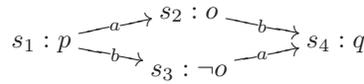
\xymatrix@R-25pt{
&{s_2:o}\ar[rd]|b&\\
{s_1:p}\ar[ur]|a\ar[dr]|b&&{s_4: q}\\
&{s_3:\neg o}\ar[ur]|a&
}
$$
Now we can also check that the operator $\U$ defined by $\Khm{\neg\psi,\top,\bot}$ is indeed an \textit{universal modality}:
$$\begin{array}{|rcl|}
\hline

\M,s\vDash \U \varphi&\Leftrightarrow& \text{ for all }t\in \S, \M, t\vDash\varphi   \\
\hline
\end{array}
$$

The following formulas are valid on all models.
	
\begin{proposition}\label{prop:validEMPKhm}
$\vDash \U(p \to q)\to \Khm{p,\bot,q}$
\end{proposition}
\begin{proof}
Assuming that $\M,s\vDash\U(p\to q)$, it means that $\M,t\vDash p\to q$ for all $t\in\calS$. Given $\M,t\vDash p$, it follows that $\M,t\vDash q$. Thus, we have $\epsilon$ is strongly $\bot$-executable at $t$. Therefore, we have $\M,s\vDash\Khm{p,\bot,q}$. 
\end{proof}

\begin{proposition}
$\vDash \Khm{p,o,r}\land\Khm{r,o,q}\land\U(r\to o)\to \Khm{p,o,q}$
\end{proposition}
\begin{proof}
Assuming $\M,s\vDash \Khm{p,o,r}\land\Khm{r,o,q}\land\U(r\to o)$, we will show that $\M,s\vDash \Khm{p,o,q}$.
Since $\M,s\vDash \Khm{p,o,r}$, it follows that there exists $\sigma\in\Act^*$ such that for each $\M,u\vDash p$, $\sigma$ is strongly $o$-executable at $u$ and that $\M,v\vDash r$ for each $v$ with $u\rel{\sigma}v$. Since $\M,s\vDash \Khm{r,o,q}$, it follows that there exists $\sigma'\in\Act^*$ such that for each $\M,v'\vDash r$, $\sigma'$ is strongly $o$-executable at $v'$ and that $\M,t\vDash q$ for each $t$ with $v'\rel{\sigma}t$. In order to show $\M,s\vDash \Khm{p,o,q}$, we only need to show that $\sigma\sigma'$ is strongly $o$-executable at $u$ and that $\M,t'\vDash q$ for each $t'$ with  $u\rel{\sigma\sigma'}t'$, where $u$ is a state with $\M,u\vDash p$.

By assumption, we know that $\sigma$ is strongly $o$-executable at $u$, and for each $v$ with $u\rel{\sigma}v$, it follows by assumption that $\M,v\vDash r$ and $\sigma'$ is strongly $o$-executable at $v$. Moreover, since $\M,s\vDash\U(r\to o) $, it follows that $\M,v\vDash o$ for each $v$ with $u\rel{\sigma}v$. Thus, $\sigma\sigma'$ is strongly $o$-executable at $u$. What is more, for each $t'$ with $u\rel{\sigma\sigma'}t'$, there is $v$ such that $u\rel{\sigma}v\rel{\sigma'}t'$ and $\M,v\vDash r$, it follows by assumption that $\M,t'\vDash q$. Therefore, we have $\M,s\vDash \Khm{p,o,q}$. 
\end{proof}
\begin{proposition}
$\vDash \Khm{p,o,q}\land\neg\Khm{p,\bot,q}\to \Khm{p,\bot,o}$
\end{proposition}
\begin{proof}
Assuming $\M,s\vDash \Khm{p,o,q}\land\neg\Khm{p,\bot,q}$, we will show that $\M,s\vDash\Khm{p,\bot,o}$. 
Since $\M,s\vDash\Khm{p,o,q}$, it follows that there exists $\sigma\in\Act^*$ such that for each $\M,u\vDash p$, $\sigma$ is strongly $o$-executable at $u$ and $\M,v\vDash q$ for all $v$ with $u\rel{\sigma}v$. If $\sigma\in\Act\cup\{\epsilon\}$, it follows that $\M,s\vDash\Khm{p,\bot,q}$. Since $\M,s\vDash\neg\Khm{p,\bot,q}$, it follows that $\sigma\not\in\Act\cup\{\epsilon\}$. Thus, $\sigma=a_1\cdots a_n$ where $n\geq 2$. Let $u$ be a state such that $\M,u\vDash p$. Since $\sigma=a_1\cdots a_n$ is strongly $o$-executable at $u$, it follows that $a_1$ is executable at $u$. Moreover, since $n\geq 2$, we have $\M,v\vDash o$ for each $v$ with $u\rel{a_1}v$. Therefore, we have $\M,s\vDash\Khm{p,\bot,o}$.
\end{proof}


\begin{proposition}\label{prop:validREPKhm}
$\vDash \U(p'\to p)\land\U(o\to o')\land\U(q\to q')\land\Khm{p,o,q}\to \Khm{p',o',q'}$
\end{proposition} 
\begin{proof}
Assuming $\M,s\vDash \U(p'\to p)\land\U(o\to o')\land\U(q\to q')\land\Khm{p,o,q}$, we will show $\M,s\vDash\Khm{p',o',q'}$. Since $\M,s\vDash\Khm{p,o,q}$, it follows that there exists $\sigma\in\Act^*$ such that for each $\M,u\vDash p$: $\sigma$ is strongly $o$-executable at $u$ and $\M,v\vDash q$ for each $v$ with $u\rel{\sigma}v$. Let $s'$ be a state with $\M,s'\vDash p'$. Next we will show that $\sigma$ is strongly $o'$-executable at $s'$ and $\M,v'\vDash q'$ for all $v'$ with $s'\rel{\sigma}v'$.

Since $\M,s\vDash\U(p'\to p)$, it follows that $\M,s'\vDash p$. Thus, $\sigma$ is strongly $o$-executable at $s'$ and $\M,v'\vDash q$ for each $v'$ with $s'\rel{\sigma} v'$. Since $\M,s\vDash \U(o\to o')$, it follows that $\sigma$ is strongly $o'$-executable at $s'$. Since $\M,s\vDash\U(q\to q')$, it follows that $\M,v'\vDash q'$ for each $v'$ with $s'\rel{\sigma} v'$.
\end{proof}
    \begin{definition}[Deductive System $\SKhm$]
		The axioms and rules shown in Table~\ref{tab:SysSKhm} constitutes the proof system $\SKhm$. 
	\end{definition}
Note that $\DISTU,\NECU, \AxTrU$ are standard for the universal modality $\U$. $\AxTransKU$ and $\AxTransKU$ are introspection axioms reflecting that $\Khmm$ formulas are global. $\EMPKhm$ captures the interaction between $\U$ and $\Khmm$ via empty plan. $\COMP$ is the new composition axiom for $\Khmm$.  $\AxRepKhm$ shows how we can weaken the knowing how claims. $\KHDisjunction$ is the characteristic axiom for $\SKhm$ compared to the system for binary $\mathcal{K}h$, and it  expresses the condition for the necessity of the intermediate steps.  
\begin{table}[htbp]	
\centering
			\begin{tabular}{|lc|}
            \hline
				\multicolumn{2}{|l|}{\textbf{Axioms}}\\[1ex] 
				\TAUT & \text{all tautologies of propositional logic}\\[1ex] 
				\DISTU & $\U p\land\U (p\to q)\to \U q$\\[1ex] 
				\AxTrU& $\U p\to p $\\[1ex] 
				\AxTransKU& $\Khm{p, o,q}\to\U\Khm{p, o,q}$\\[1ex] 
				\AxEucKU& $\neg \Khm{p, o,q}\to\U\neg\Khm{p, o,q}$\\[1ex] 
				\EMPKhm &$\U(p \to q)\to \Khm{p,\bot,q}$ \\[1ex]
				\COMP & $\Khm{p,o,r}\land\Khm{r,o,q}\land\U(r\to o)\to \Khm{p,o,q}$\\[1ex] 
				\KHDisjunction& $\Khm{p,o,q}\land\neg\Khm{p,\bot,q}\to \Khm{p,\bot,o}$\\[1ex]
				\AxRepKhm & $\U(p'\to p)\land\U(o\to o')\land\U(q\to q')\land\Khm{p,o,q}\to \Khm{p',o',q'}$ \\[1ex]
				\multicolumn{2}{|l|}{\textbf{Rules}}\\[1ex]
				\multicolumn{2}{|c|}{
					\begin{tabular}{cccccccc}
						\MP & $\dfrac{\varphi,\varphi\to\psi}{\psi}$ &~~~~~~~~ &\NECU &$\dfrac{\varphi}{\U\varphi}$&~~~~~~~~ &\SUB &$\dfrac{\varphi(p)}{\varphi[\psi\slash p]}$
					\end{tabular}	
				}\\
                \hline
			\end{tabular}            
			\caption{System $\SKhm$}\label{tab:SysSKhm}
\end{table}
\begin{remark}
Note that the corresponding axioms for $\COMP$, $\EMPKhm$ and $\AxRepKhm$ in the setting of binary $\mathcal{K}h$ are the following:\footnote{We can obtain the corresponding axioms by taking the intermediate constraint as $\top$. Note that in \cite{Wang15lori}, we use the name \texttt{WKKh} for $\WSKh$.}
\begin{center}
			\begin{tabular}{|c|c|}
            \hline
           \COMPKh & $\Kh{p,q}\land\Kh{q,r}\to \Kh{p,r}$\\
           \EMP & $\U(p\to q)\to \Kh{p,q}$\\
        \WSKh & $\U(p'\to p) \land \U(q\to q')\land \Kh{p, q}\to \Kh{p', q'}$\\
		    \hline
             \end{tabular}
\end{center}
In the system $\SKh$ of \cite{Wang15lori} \WSKh\ can be derived using $\COMPKh$ and $\EMP$. However, $\AxRepKhm$ cannot be derived using $\COMP$ and $\EMPKhm$. In particular, $\Khm{p',\bot, p}\land\Khm{p,o,q}\to \Khm{p',o,q}$ is not valid due to the lack of $\U(p\to o)$, in contrast with the $\SKh$-derivable $\Kh{p', p}\land \Kh{p, q}\to \Kh{p', q}$ which is  crucial in the derivation of \WSKh\ in $\SKh$. 
\end{remark}

Since $\U$ is an universal modality, \DISTU\ and \AxTrU\ are obviously valid. Due to the fact that the modality $\Khmm$ is not local, it is easy to show that \AxTransKU\ and \AxEucKU\ are valid. Moreover, by Propositions \ref{prop:validEMPKhm}--\ref{prop:validREPKhm}, we have that all axioms are valid. Due to a standard argument in modal logic, we know that the rules \MP, \NECU\ and \SUB\ preserve formula's validity. The soundness of $\SKhm$ follows immediately.
\begin{theorem}
	$\SKhm$ is sound w.r.t. the class of all models. 
\end{theorem}

Below we derive some theorems and rules that are useful in the later proofs.
\begin{proposition}
We can derive the following in $\SKhm$:
\begin{center}
			\begin{tabular}{|c|c|}
				\hline
				\AxTransU&$\U p\to\U\U p$\\
				\AxEucU&$\neg\U p\to \U\neg\U p$\\
                \ThKhmL& $\U(p'\to p)\land\Khm{p,o,q}\to \Khm{p',o,q}$\\
                \ThKhmM& $\U(o\to o')\land\Khm{p,o,q}\to \Khm{p,o',q}$\\  
             \ThKhmR& $\U(q\to q')\land\Khm{p,o,q'}\to \Khm{p,o,q'}$\\
                             \Univers &   $\U\neg p\to\Khm{p,\bot,\bot}$\\
\REU & from $\phi\lra \psi$ prove $\U \phi\lra \U\psi$  \\
\RE & from $\phi\lra \psi$ prove $\chi\lra \chi'$ \\& where $\chi'$ is obtained by replacing some occurrences of $\phi$ in $\chi$ by $\psi$. \\

				\hline
			\end{tabular}
		\end{center}
\end{proposition}
\begin{proof}\REU\ is immediate given $\DISTU$ and $\NECU$. 
$\AxTransU$ and $\AxEucU$ are special cases of $\AxTransKU$ and $\AxEucKU$ respectively. $\ThKhmL, \ThKhmM, \ThKhmR$ are the special cases of $\AxRepKhm$. To prove \Univers, first note that $\U\neg p\lra \U (p\to\bot)$ due to \REU. Then due to $\EMPKhm$, we have $\U\neg p\to \Khm{p, \bot, \bot}$. $\RE$ can be obtained by using $\AxRepKhm$ and $\NECU$.   

\end{proof}
	\section{Completeness}\label{Sec.Proof}
    This section will prove that $\SKhm$ is complete w.r.t. the class of all models. The key is to build a canonical model based on a fixed maximal consistent set, just as in \cite{Wang15lori}. However, the canonical model here is much more complicated. Firstly, the state of the canonical model is a pair consisting of a maximal consistent set and a marker which will play an important role in defining the witness plan for $\Khmm$-formulas. Secondly, different from the canonical model in \cite{Wang15lori} where each formula of the form $\Kh{\psi,\phi}$ is realized by an one-step witness plan, some $\Khm{\psi,\chi,\phi}$ formulas here have to be realized by a two-step witness plan, and the intermediate states need to satisfy $\chi$.
    
	Here are some notions before we prove the completeness. Given a set of $\LKhm$ formulas $\Delta$, let $\Delta|_{\Khmm}$ and $\Delta|_{\neg\Khmm}$ be the collections of its positive and negative $\Khmm$ formulas:  
	$$\Delta|_{\Khmm}=\{\theta \mid\theta=\Khm{\psi,\chi,\varphi}\in\Delta  \};$$
    $$\Delta|_{\neg\Khmm}=\{\theta \mid\theta=\neg\Khm{\psi,\chi,\varphi}\in\Delta  \}.$$	
	In the following, let $\Gamma$ be a maximal consistent set (MCS) of $\LKhm$ formulas. We first prepare ourselves with some handy propositions. 
	\begin{definition}
		Let $\MCS$ be the set of all MCS $\Delta$ such that $\Delta|_\Khmm=\Gamma|_\Khmm$. 
	\end{definition}
Since every $\Delta\in \MCS$ is maximal consistent it follows immediately that: 
	\begin{proposition}\label{prop:ShareKHow}
		For each $\Delta\in\MCS$, we have $\Khm{\psi,\chi,\phi}\in\Gamma$ if and only if $\Khm{\psi,\chi,\phi}\in\Delta$ for all $\Khm{\psi,\chi,\phi}\in \LKhm$.
	\end{proposition}
	\begin{proposition}\label{prop:allphiImplyUphi}
		If $\phi\in\Delta$ for all $\Delta\in\MCS$ then $\U\phi\in\Delta$ for all $\Delta\in\MCS$.
	\end{proposition}
	\begin{proof}
		Suppose $\varphi\in \Delta$ for all $\Delta\in \Phi_\Gamma$, then by the definition of $\Phi_\Gamma$, $\neg\varphi$ is not consistent with $\Gamma|_{\Khmm}\cup\Gamma|_{\neg\Khmm}$, for otherwise $\Gamma|_{\Khmm}\cup\Gamma|_{\neg\Khmm}\cup\{\neg \varphi\}$ can be extended into a maximal consistent set in $\Phi_\Gamma$ due to a standard Lindenbaum-like argument. Thus there are $\Khm{\psi_1,\chi_1, \varphi_1}$, \dots, $\Khm{\psi_k,\chi_k, \varphi_k} \in \Gamma|_{\Khmm}$ and $\neg \Khm{\psi'_1,\chi'_1, \varphi'_1}$, \dots, $\neg\Khm{\psi'_l,\chi'_l, \varphi'_l} \in \Gamma|_{\neg \Khmm}$ such that 
		$$\vdash \bigwedge_{1\leq i\leq k}\Khm{\psi_i,\chi_i, \varphi_i}\land \bigwedge_{1\leq j\leq l}\neg \Khm{\psi'_j,\chi'_j, \varphi'_j}\to \varphi.$$
		By $\NECU$,  $$\vdash \U(\bigwedge_{1\leq i\leq k}\Khm{\psi_i,\chi_i, \varphi_i}\land \bigwedge_{1\leq j\leq l}\neg \Khm{\psi'_j,\chi'_j, \varphi'_j}\to \varphi).$$ By $\DISTU$  we have: $$\vdash\U(\bigwedge_{1\leq i\leq k}\Khm{\psi_i,\chi_i, \varphi_i}\land \bigwedge_{1\leq j\leq l}\neg \Khm{\psi'_j,\chi'_j, \varphi'_j}) \to \U\varphi.$$   Since $\Khm{\psi_1,\chi_1, \varphi_1}$, \dots, $\Khm{\psi_k,\chi_k, \varphi_k} \in \Gamma$, we have $\U\Khm{\psi_1,\chi_1, \varphi_1}$, \dots, $\U\Khm{\psi_k,\chi_k, \varphi_k}\in \Gamma$ due to $\AxTransKU$ and the fact that $\Gamma$ is a maximal consistent set. Similarly, we have $\U\neg \Khm{\psi'_1,\chi'_1, \varphi'_1}$, \dots, $\U\neg\Khm{\psi'_l,\chi'_l, \varphi'_l} \in \Gamma$ due to $\AxEucKU$. By \DISTU\ and \NECU, it is easy to show that $\vdash \U(p\land q)\lra\U p\land \U q$. Then due to a slight generalization, we have: $$\U(\bigwedge_{1\leq i\leq k}\Khm{\psi_i,\chi_i, \varphi_i}\land \bigwedge_{1\leq j\leq l}\neg \Khm{\psi'_j,\chi'_j, \varphi'_j}) \in \Gamma.$$  Now it is immediate that $\U\varphi\in\Gamma$. Due to Proposition~\ref{prop:ShareKHow}, $\U\varphi\in\Delta$ for all $\Delta\in \Phi_\Gamma.$
	\end{proof}
	
	\begin{proposition}\label{prop:OneStepExistence}
		Given $\Khm{\psi,\top,\phi}\in \Gamma$ and $\Delta\in\Phi_\Gamma$, if $\psi\in\Delta$ then there exists $\Delta'\in\Phi_\Gamma$ such that $\phi\in\Delta'$.
	\end{proposition}
	\begin{proof}
Assuming $\Khm{\psi,\top,\phi}\in \Gamma$ and $ \psi\in\Delta\in \Phi_\Gamma$, if there does not exist $\Delta'\in\Phi_\Gamma$ such that $\phi\in\Delta'$, it means that $\neg\phi\in\Delta'$ for all $\Delta'\in\Phi_\Gamma$. It follows by Proposition \ref{prop:allphiImplyUphi} that $\U\neg\phi\in\Gamma$, namely $\Khm{\phi,\top,\bot}\in\Gamma$. Since $\U(\phi\to \bot)$ and $\Khm{\psi,\top,\phi}\in\Gamma$, it follows by \COMP\ that $\Khm{\psi,\top,\bot}\in\Gamma$ namely, $\U\neg\psi\in\Gamma$. By Proposition \ref{prop:ShareKHow}, we have that $\U\neg\psi\in\Delta$. It follows by \AxTrU\ that $\neg\psi\in\Delta$. This is contradictory with $\psi\in\Delta$. Therefore, there exists $\Delta'\in\Phi_\Gamma$ such that $\phi\in\Delta'$.
	\end{proof}
	
	\begin{definition} Let the set of action symbols $\Act_\Gamma$ be defined as  $\Act_\Gamma=\{\lr{\psi,\bot,\phi}\mid \Khm{\psi,\bot,\phi}\in\Gamma \}\cup\{\lr{\chi^\psi,\phi}\mid \Khm{\psi,\chi,\phi},\neg\Khm{\psi,\bot,\phi}\in\Gamma \}$. 
	\end{definition}
The later part of $\Act_\Gamma$ is to handle the cases where the intermediate state is indeed necessary: $\neg\Khm{\psi,\bot,\phi}$ makes sure that you cannot have a plan to guarantee $\phi$ in less than two steps.

\medskip
   
In the following we build a separate canonical model for each MCS $\Gamma$, for it is not possible to satisfy all of $\Khmm$ formulas simultaneously in a single model since they are global. Because the later proofs are quite technical, it is very important to first understand the ideas behind the canonical model construction. Note that to satisfy a $\Khm{\psi, \chi, \phi}$ formula, there are two cases to be considered: 

(1) $\Khm{\psi, \bot, \phi}$ holds and we just need an one-step witness plan, which can be handled similarly using the techniques developed in \cite{Wang15lori};

(2) $\Khm{\psi, \bot, \phi}$ does not hold, and we need to have a witness plan which at least involves an intermediate $\chi$-stage. By \KHDisjunction, $\Khm{\psi, \bot, \chi}$ holds. It is then tempting to reduce  $\Khm{\psi, \chi, \phi}$ to $\Khm{\psi, \bot, \chi}\land \Khm{\chi, \chi, \phi}$. However, it is not correct since we may not have a strongly $\chi$-executable plan to make sure $\phi$ from \textit{any} $\chi$-state. Note that $\Khm{\psi, \chi, \phi}$ and $\Khm{\psi, \bot, \chi}$ only make sure we can start from \textit{certain} $\chi$-states that result from the witness plan for $\Khm{\psi, \bot, \chi}$. However, we cannot refer to such $\chi$-states in the language of $\LKhm$. This is why we include $\chi^\psi$ markers in the building blocks of the canonical model besides maximal consistent set. $\chi^\psi$ roughly tells us where does this state ``comes from''. \footnote{In \cite{Wang15lori}, the canonical models are much simpler: we just need MCSs and the canonical relations are simply labeled by $\lr{\psi, \phi}$ for $\Kh{\psi,\phi}\in \Gamma$. }


	\begin{definition}[Canonical Model]
		The canonical model for $\Gamma$ is a tuple $\calM^c_{\Gamma}=\lr{\calS^c,\calR^c,\calV^c}$ where:
		\begin{itemize}
			\item $\calS^c=\{(\Delta,\chi^\psi)\mid \chi\in\Delta\in\MCS $, and $\lr{\chi^\psi,\phi}\in\Act_\Gamma$ for some $\phi$ or $\lr{\psi,\bot,\chi}\in\Act_\Gamma \}$. We write the pair in $\calS$ as $w,v,\cdots$, and refer to the first entry of $w\in\calS$ as $\Lef{w}$, to the second entry as $\Rig{w}$;
			\item $w\rel{\lr{\psi,\bot,\phi}}_c w'$ iff $\psi\in\Lef{w}$ and $\Rig{w'}=\phi^\psi$;
			\item $w\rel{\lr{\chi^\psi,\phi}}_c w'$ iff $\Rig{w}=\chi^\psi$ and $\phi\in\Lef{w'}$;
			\item $p\in\calV^c (w)$ iff $p\in\Lef{w}$.
\end{itemize}
For each $w\in\calS$, we also call $w$ a $\psi$-state if $\psi\in\Lef{w}$.
\end{definition}

In the above definition, $\Rig{w}$ marks the use of $w$ as an intermediate state.  The same maximal consistent set $\Delta$ may have different uses depending on different $\Rig{w}$. We will make use of the transitions $w\rel{\lr{\psi,\bot,\chi}}_c v\rel{\lr{\chi^\psi,\phi}}_c w'$ where $\Rig{v}=\chi^\psi$. Note that if $\Rig{w}=\chi^\psi$ then $w\rel{\lr{\chi^\psi,\phi}}_cv$ for each $\phi$-state $v$. The highly non-trivial  part of the later proof of the truth lemma is to show adding such transitions and making them to be composed arbitrarily will not cause some $\Khm{\psi,\chi,\phi}\not\in \Lef{w}$ to hold at $w$.  

\medskip

We first show that each $\Delta\in \Phi_\Gamma$ appears as $\Lef{w}$ for some $w\in \calS^c$.
	\begin{proposition}\label{prop:allMCSinS}
		For each $\Delta\in\MCS$, there exists $w\in\calS^c$ such that $\Lef{w}=\Delta$.
	\end{proposition}
	\begin{proof}
		Since $\vdash\top\to\top$, it follows by \NECU\ that $\vdash\U(\top\to\top)$. Thus, we have $\U(\top\to\top)\in\Gamma$. It follows by \EMPKhm\ that $\Khm{\top,\bot,\top}\in\Gamma$. It follows that $a=\lr{\top,\bot,\top}\in\Act_\Gamma$. Since $\top\in\Delta$, it follows that $(\Delta,\top^\top)\in\calS^c$.
	\end{proof}
	Since $\Gamma\in\Phi_\Gamma$, it follows by Proposition \ref{prop:allMCSinS} that $\calS^c\neq\emptyset$.
\medskip
    
Proposition \ref{prop:allphiImplyUphi} helps us to prove the following two handy propositions which will play crucial roles in the completeness proof. Note that according to Proposition \ref{prop:allphiImplyUphi}, to obtain that $\U\phi$ in all the $\Delta\in \Phi_\Gamma$, we just need to show that $\phi$ is in all the $\Delta\in \Phi_\Gamma$, not necessarily in all the $w\in \calS^c$. 
    
	\begin{proposition}\label{prop:IMPinHead}
		Given $a=\lr{\psi',\bot,\phi'}\in\Act_\Gamma$, If for each $\psi$-state $w\in\calS^c$ we have that $a$ is executable at $w$, then $\U(\psi\to\psi')\in \Gamma$.
	\end{proposition}
	\begin{proof}
		Suppose that every $\psi$-state has an outgoing $a$-transition,  then by the definition of $\R^c$, $\psi'$ is in all the $\psi$-states. For each $\Delta\in\Phi_\Gamma$, either $\psi\not\in \Delta$, or $\psi\in\Delta$ thus $\psi'\in \Delta$. Now by the fact that $\Delta$ is maximally consistent it is not hard to show $\psi\to\psi'\in\Delta$ in both cases.  By Proposition~\ref{prop:allphiImplyUphi}, $\U(\psi\to\psi') \in \Delta$ for all $\Delta\in\Phi_\Gamma.$ It follows by $\Gamma\in\Phi_\Gamma$ that $\U(\psi\to\psi') \in\Gamma$.
	\end{proof}
	
	\begin{proposition}\label{prop:IMPinTail}
		Given $w\in\calS^c$ and $a=\lr{\psi,\bot,\phi'}$ or $\lr{\chi^{\psi},\phi'}\in\Act_\Gamma$ such that $a$ is executable at $w$, if $\phi\in\Lef{w'}$ for each $w'$ with $w\rel{a}w'$ then $\U(\phi'\to\phi)\in\Gamma$.
	\end{proposition}
	\begin{proof}
		Firstly, we focus on the case of $a=\lr{\psi,\bot,\phi'}$.
		For each $\Delta\in\Phi_\Gamma$ with $\phi'\in \Delta$, we have $v=(\Delta,\phi'^\psi)\in\calS^c$. Since $\lr{\psi,\bot,\phi'}$ is executable at $w$, it means that $\psi\in\Lef{w}$. By the definition, it follows that $w\rel{a}v$. Since $\phi\in\Lef{w'}$ for each $w'$ with $w\rel{a}w'$, it follows that $\phi\in\Lef{v}$. Therefore, we have $\phi\in\Delta$ for each $\Delta\in\Phi_\Gamma$ with $\phi'\in\Delta$, namely $\phi'\to\phi\in \Delta$ for all $\Delta\in\Phi_\Gamma$. It follows by Proposition \ref{prop:allphiImplyUphi} that $\U(\phi'\to\phi)\in\Gamma$.

		Secondly, we focus on the case of $a=\lr{\chi^\psi,\phi'}$. For each $\Delta\in\Phi_\Gamma$ with $\phi'\in \Delta$, it follows by Proposition \ref{prop:allMCSinS} that there exists $v\in\calS^c$ such that $\Lef{v}=\Delta$. Since $a$ is executable at $w$, it follows that $w\rel{a}v$. Since $\phi\in\Lef{w'}$ for each $w'$ with $w\rel{a}w'$, it follows that $\phi\in\Lef{v}$. Therefore, we have shown that $\phi'\in\Delta$ implies $\phi\in\Delta$ for all $\Delta\in\Phi_\Gamma$. It follows by Proposition \ref{prop:allphiImplyUphi} that $\U(\phi'\to\phi)\in\Gamma$.
	\end{proof}
\indent Before proving the truth lemma, we first need a handy result.  
\begin{proposition}\label{prop:phi_n}
Given a non-empty sequence $\sigma=a_1\cdots a_n\in\Act^*_\Gamma$ where $a_i=\lr{\psi_i,\bot,\phi_i}$ or $a_i=\lr{\chi^{\psi_i}_i,\phi_i}$ for each $1\leq i\leq n$, we have $\Khm{\psi,\chi,\phi_i}\in\Gamma$ for all $1\leq i\leq n $ if for each $\psi$-states $w\in \calS^c$:
\begin{itemize}
\item $\sigma$ is strongly executable at $w$;
\item $w\rel{\sigma_j}t'$ implies $\chi\in\Lef{t'}$ for all $1\leq j <n$.
\end{itemize} 
\end{proposition}
\begin{proof}
		If there is no $\psi$-state in $\calS^c$, it follows that $\neg\psi\in \Lef{w'}$ for each $w'\in\calS^c$. It follows by Proposition \ref{prop:allMCSinS} that $\neg\psi\in\Delta$ for all $\Delta\in\Phi_\Gamma$. By Proposition \ref{prop:allphiImplyUphi}, we have $\U\neg\psi\in \Gamma$. By \Univers, $\Khm{\psi,\bot,\bot}\in \Gamma$. Since $\vdash \bot\to\chi$ and $\vdash\bot\to \phi$. Then by \NECU, we have $\vdash\U(\bot\to\chi)$ and $\vdash\U(\bot\to\phi)$. By $\ThKhmM$ and $\ThKhmR$, it is obvious that $\Khm{\psi,\chi,\phi}\in\Gamma$.
		
		Next, assuming $v\in\calS^c$ is a $\psi$-state, we will show $\Khm{\psi,\chi,\phi}\in\Gamma$.
There are two cases: $n=1$ or $n\geq 2$. For the case of $n=1$, we will prove it directly; for the case of $n\geq 2$, we will prove it by induction on $i$.
			\begin{itemize}
				\item $n=1$.  If $a_1$ is in the form of $\lr{\chi_1^{\psi_1},\phi_1}$, by the definition of $\rel{\lr{\chi_1^{\psi_1},\phi_1}}$ it follows that $\Rig{w}=\chi_1^{\psi_1}$ for each $\psi$-state $w$. Let $\chi_0$ be a formula satisfying that $\vdash\chi_0\lra\chi_1$ and $\chi_0\neq \chi_1$. By the rule of Replacement of Equals \RE, it follows that $\lr{\chi_0^{\psi_1},\phi_1}\in\Act_\Gamma$. Let $w'=(\Lef{v},\chi_0^{\psi_1})$  then it follows that $w'\in\calS^c$. Since $\psi\in\Lef{v}$, then we have $\psi\in\Lef{w'}$. However, since $\Rig{w'}=\chi_1^{\psi_1}\not=\chi_0^{\psi_1}$, $\sigma=\lr{\chi_1^{\psi_1},\phi_1}$ is not executable at the $\psi$-state $w'$, contradicting the assumption that $\sigma$ is strongly executable at all $\psi$-states. Therefore, we know that $a_1$ cannot be in the form of $\lr{\chi_1^{\psi_1},\phi_1}$. 
                
                \medskip
				If $a_1=\lr{\psi_1,\bot,\phi_1}$, it follows that $\Khm{\psi_1,\bot,\phi_1}\in\Gamma$. Since $a_1$ is executable at each $\psi$-state, it follows by Proposition \ref{prop:IMPinHead} that $\U(\psi\to\psi_1)\in\Gamma$. 
				Since $\Khm{\psi_1,\bot,\phi_1}\in\Gamma$, it follows by \ThKhmL\ that $\Khm{\psi,\bot,\phi_1}\in\Gamma$. By \NECU\ and \ThKhmM, it is clear that  $\Khm{\psi,\chi,\phi_1}\in\Gamma$. 
				
\item $n\geq 2$.			
By induction on $i$, next we will show that 					 $\Khm{\psi,\chi,\phi_i}\in\Gamma$ 
for each $1\leq i\leq  n$.
 For the case of $i=1$, with the similar proof as in the case of $n=1$, we can show that $a_1$ can only be $\lr{\psi_1,\bot,\phi_1}$ and $\U(\psi\to\psi_1)\in\Gamma$. Therefore by $\AxRepKhm$ we have $\Khm{\psi,\chi,\phi_1}\in\Gamma$.	
					Under the induction hypothesis (IH) that $\Khm{\psi,\chi,\phi_i}\in\Gamma$ for each $1\leq i\leq  k$, we will show that $\Khm{\psi,\chi,\phi_{k+1}}\in\Gamma$, where $1\leq k\leq n-1$.
					Because $\sigma$ is strongly executable at $v$, it follows that there are $w',v'\in\calS^c$ such that
					\[\xymatrix{
						v\ar[r]^{a_1}&\cdots\ar[r]^{a_{k-1}}&w'\ar[r]^{a_k}&v'\ar[r]^{a_{k+1}}&\cdots\ar[r]^{a_n}&t.\\
					} \]
					Moreover, for each $t'$ with $w'\rel{a_{k}}t'$ we have  $\chi\in\Lef{t'}$. It follows by Proposition \ref{prop:IMPinTail} that $\U(\phi_{k}\to \chi)\in\Gamma \ (\blacktriangle)$. Proceeding, there are two cases of $a_{k+1}$:
					\begin{itemize}
						\item $a_{k+1}=\lr{\psi_{k+1},\bot,\phi_{k+1}}$. Since $\sigma$ is strongly executable at $v$, it follows that for each $t'$ with $w'\rel{a_{k}}t'$ we know that $a_{k+1}$ is executable at each $t'$. It follows by the definition of $\rel{\lr{\psi_{k+1},\bot,\phi_{k+1}}}$ that $\psi_{k+1}\in \Lef{t'}$. Moreover, since $a_k$ is executable at $w'$, it follows by Proposition \ref{prop:IMPinTail} that $\U(\phi_{k}\to\psi_{k+1})\in\Gamma$. Since $a_{k+1}\in\Act_\Gamma$, it then follows that $\Khm{\psi_{k+1},\bot,\phi_{k+1}}\in\Gamma$. It then follows by \ThKhmL\ that $\Khm{\phi_{k},\bot,\phi_{k+1}}\in\Gamma$. Since $\vdash\U(\bot\to\chi)$, it follows by \ThKhmM\ that $\Khm{\phi_{k},\chi,\phi_{k+1}}\in\Gamma$. Since by IH we have that $\Khm{\psi,\chi,\phi_{k}}\in\Gamma$, It follows from $(\blacktriangle)$ and \COMP\ that $\Khm{\psi,\chi,\phi_{k+1}}\in\Gamma$.
						\item $a_{k+1}=\lr{\chi_{k+1}^{\psi_{k+1}},\phi_{k+1}}$.
						Since $\sigma$ is strongly executable at $v$, it follows that for each $t'$ with $w'\rel{a_{k}}t'$ we know that $a_{k+1}$ is executable at $t'$. Then we have that $\Rig{t'}=\chi_{k+1}^{\psi_{k+1}}$ for each $t'$ with $w'\rel{a_{k}}t'$. \\
                        Note that the action $a_{k}$ cannot be in the form of $\lr{\chi_{k}^{\psi_{k}},\phi_{k}}$. Suppose it can be, let $v''=(\Lef{v'},\chi_0^{\psi_{k+1}})$ where $\vdash\chi_0\lra \chi_{k+1}$ and $\chi_0\neq \chi_{k+1}$. Since $w'\rel{a_k}v'$, it follows that $\phi_k\in\Lef{v'}$. Then it follows by the definition of transitions that $w'\rel{a_{k}}v''$. However, we know that $\Rig{v''}\neq \chi_{k+1}^{\psi_{k+1}}$ thus $a_{k+1}=\lr{\chi_{k+1}^{\psi_{k+1}},\phi_{k+1}}$ is not executable at $v''$, contradicting the strong executability. Therefore, we know that $a_{k}$ cannot be in the form of $\lr{\chi_{k}^{\psi_{k}},\phi_{k}}$.
						
						Now $a_{k}=\lr{\psi_{k},\bot,\phi_{k}}$. Since $w'\rel{a_{k}}v'$ and $a_{k+1}=\lr{\chi_{k+1}^{\psi_{k+1}},\phi_{k+1}}$ is executable at $v'$, we have $\Rig{v'}=\phi_k^{\psi_k}=\chi_{k+1}^{\psi_{k+1}}$ by definition of transitions. It follows that $\psi_{k}=\psi_{k+1}$ and $\phi_{k}=\chi_{k+1}$. Since $a_{k+1}\in\Act_\Gamma$, it follows that $\Khm{\psi_{k+1},\chi_{k+1},\phi_{k+1}}\in\Gamma$. Thus, we have $\Khm{\psi_{k},\phi_{k},\phi_{k+1}}\in\Gamma$. 
						By $(\blacktriangle)$ and \ThKhmM\ we then have that $\Khm{\psi_{k},\chi,\phi_{k+1}}\in\Gamma \ (\blacktriangledown)$. If $k=1$, by Proposition~\ref{prop:IMPinHead} it is easy to show that $\U(\psi\to\psi_1)\in\Gamma$. Then by \ThKhmL\ we have $\Khm{\psi,\chi,\phi_{k+1}}\in\Gamma$. If $k>1$, there is a state $w''$ such that
						\[\xymatrix{
							v\ar[r]^{a_1}&\cdots\ar[r]^{a_{k-2}}&w''\ar[r]^{a_{k-1}}&w'\ar[r]^{a_k}&v'\ar[r]^{a_{k+1}}&\cdots\ar[r]^{a_n}&t.\\
						} \]
						Since $\sigma$ is strongly executable at $v$, it follows that for each $t'$ with $w''\rel{a_{k-1}}t'$ we have $a_{k}$ is executable at $t'$. It follows by the definition of $\rel{\lr{\psi_k,\bot,\phi_k}}$, it follows that  $\psi_k\in\Lef{t'}$ for each $t'$ with $w''\rel{a_{k-1}}t'$.
						Since $a_{k-1}$ is executable at $w''$, it follows by Proposition \ref{prop:IMPinTail} that $\U(\phi_{k-1}\to \psi_{k})\in \Gamma$.Moreover, since $v\rel{\sigma_{k-1}}t'$ for each $t' $ with $w''\rel{a_{k-1}}t'$, it follows that $\chi\in\Lef{t'}$. Thus by Proposition \ref{prop:IMPinTail} again, we have $\U(\phi_{k-1}\to \chi)\in \Gamma$. Since we have proved $(\blacktriangledown)$, it follows by \ThKhmL\ that $\Khm{\phi_{k-1},\chi,\phi_{k+1}}\in\Gamma$. Since by IH we have $\Khm{\psi,\chi,\phi_{k-1}}\in\Gamma$, it follows by \COMP\ that $\Khm{\psi,\chi,\phi_{k+1}}\in\Gamma$. 
					\end{itemize}				
\end{itemize}
\end{proof}	
Now we are ready to prove the truth lemma. 	
	\begin{lemma}
		For each $\phi$, we have $\calM^c_\Gamma,w\vDash\phi$ iff $\phi\in \Lef{w}$.
	\end{lemma}
	\begin{proof}
		Boolean cases are trivial, and we only focus on the case of $\Khm{\psi,\chi,\phi}$.
		
		\textbf{Left to Right:} 
		If there is no state $w'$ such that $\M^c_\Gamma,w'\vDash\psi$, it follows by induction that $\neg\psi\in \Lef{w'}$ for each $w'\in\calS^c$. It follows by Proposition \ref{prop:allMCSinS} that $\neg\psi\in\Delta$ for all $\Delta\in\Phi_\Gamma$. By Proposition \ref{prop:allphiImplyUphi}, we have $\U\neg\psi\in \Lef{w}$. By \Univers, $\Khm{\psi,\bot,\bot}\in \Lef{w}$. Since $\vdash \bot\to\chi$ and $\vdash\bot\to \phi$. Then by \NECU, we have $\vdash\U(\bot\to\chi)$ and $\vdash\U(\bot\to\phi)$. By $\ThKhmM$ and $\ThKhmR$, it is obvious that $\Khm{\psi,\chi,\phi}\in\Lef{w}$.
		
		Next, assuming $\M^c_\Gamma,v\vDash\psi$ for some $v\in\calS^c$, we will show $\Khm{\psi,\chi,\phi}\in\Lef{w}$. Since $\M^c_\Gamma,w\vDash\Khm{\psi,\chi,\phi}$, it follows that there exists $\sigma\in\Act^*$ such that for each $\M^c_\Gamma,w'\vDash\psi$: $\sigma$ is strongly $\chi$-executable at $w'$ and $\M^c_\Gamma,v'\vDash\phi$ for all $v'$ with $w'\rel{\sigma}v'$. There are two cases: $\sigma$ is empty or not.
		\begin{itemize}
			\item $\sigma=\epsilon$. This means that $\M^c_\Gamma,w'\vDash\phi$ for each $\M^c_\Gamma,w'\vDash\psi$. It follows by induction that $\psi\in\Lef{w'}$ implies $\phi\in\Lef{w'}$. Thus, we have $\psi\to\phi\in\Lef{w'}$ for all $w'\in\calS^c$. By Proposition \ref{prop:allMCSinS}, we have $\psi\to\phi\in\Delta$ for all $\Delta\in\Phi_\Gamma$. It follows by Proposition \ref{prop:allphiImplyUphi} that $\U(\psi\to\phi)\in \Lef{w}$. It then follows by \EMPKhm\ that $\Khm{\psi,\bot,\phi}\in\Lef{w}$. By \NECU\ and \ThKhmM, it is easy to show that $\Khm{\psi,\chi,\phi}\in\Lef{w}$.
			\item $\sigma=a_1\cdots a_n$ where for each $1\leq i\leq n$, $a_i=\lr{\psi_i,\bot,\phi_i}$ or $a_i=\lr{\chi^{\psi_i}_i,\phi_i}$. 
            Since $\sigma$ is strongly $\chi$-executable at each $w'$ with $\calM^c_\Gamma, w'\vDash \psi$,  it follows by IH that for each $\psi$-state $w'$: $\sigma$ is strongly executable at $w'$ and $w'\rel{\sigma_j}t'$ implies $\chi\in\Lef{t'}$ for all $1\leq j <n$.	
By Proposition \ref{prop:phi_n}, we have that $\Khm{\psi,\chi,\phi_n}\in\Lef{v}$. Since $\M^c_\Gamma,v\vDash\psi$ and $\sigma$ is strongly $\chi$-executable at $v$ and $\M^c_\Gamma,v''\vDash\phi$ for each $v''$ with $v\rel{\sigma}v''$, it follows that there exists $v'$ such that $a_n$ is executable at $v'$ and $\M^c_\Gamma,v''\vDash\phi$ for each $v''$ with $v'\rel{a_n}v''$. (Please note that $v'=v$ if $n=1$.) Note that $a_n$ is either $\lr{\psi_n, \bot, \phi_n}$ or $\lr{\chi_n^{\psi_n}, \phi_n}$. It follows by Proposition \ref{prop:IMPinTail} and IH that $\U(\phi_n\to\phi)\in\Gamma$, then we have $\U(\phi_n\to\phi)\in\Lef{v}$. It follows by \ThKhmR\ and Proposition \ref{prop:ShareKHow} that $\Khm{\psi,\chi,\phi}\in\Lef{w}$. 
\end{itemize}

This completes the proof for $w\vDash \Khm{\psi,\chi,\phi}$ implies $\Khm{\psi,\chi,\phi}\in\Lef{w}.$
\medskip
		
		\textbf{Right to Left:} Suppose that $\Khm{\psi,\chi,\phi}\in\Lef{w}$, we need to show that $\M^c_\Gamma,w\vDash\Khm{\psi,\chi,\phi}$. There are two cases: there is a state $w'\in\calS^c$ such that $\M^c_\Gamma,w'\vDash\psi$ or not. If there is no such state, it follows $\M^c_\Gamma,w\vDash\Khm{\psi,\chi,\phi}$.
		
		For the second case, let $w'$ be a state such that $\M^c_\Gamma,w'\vDash\psi$. It follows by IH that $\psi\in\Lef{w'}$. Since we already have $\Khm{\psi,\chi,\phi}\in\Lef{w}$, it follows by Proposition \ref{prop:ShareKHow} that $\Khm{\psi,\chi,\phi}\in\Gamma$. Since $\vdash\U(\chi\to\top)$, it follows by \ThKhmM\ that $\Khm{\psi,\top,\phi}\in\Gamma$. It follows by Proposition \ref{prop:OneStepExistence} that there exists $\Delta'\in\Phi_\Gamma$ such that $\phi\in\Delta'$. There are two cases: $\Khm{\psi,\bot,\phi}\in \Gamma$ or not.
		\begin{itemize}
			\item $\Khm{\psi,\bot,\phi}\in\Gamma$. It follows that $a=\lr{\psi,\bot,\phi}\in\Act_\Gamma$. Therefore, we have $v=(\Delta',\phi^\psi)\in\calS^c$. Since $\psi\in\Lef{w'}$, it follows that $w'\rel{a}v$. Thus, $a$ is strongly $\chi$-executable at $w'$. What is more, $\phi\in\Lef{v'}$ for each $v'$ with $w'\rel{a}v'$ by the definition of the transition. It follows by IH that $\M^c_\Gamma,v'\vDash\phi $ for all $v'$ with $w'\rel{a}v'$. Therefore, we have $\M^c_\Gamma,w\vDash\Khm{\psi,\chi,\phi}$ witnessed by a single step $\sigma$. 
			\item $\neg\Khm{\psi,\bot,\phi}\in\Gamma$. It follows by \KHDisjunction\ that $\Khm{\psi,\bot,\chi}\in\Gamma$. We then have $a=\lr{\psi,\bot,\chi}\in\Act_\Gamma$ and $b=\lr{\chi^\psi,\phi}\in\Act_\Gamma$. Since $\Khm{\psi,\bot,\chi}\in\Gamma$ and $\vdash\U(\bot\to\top)$, it follows by \ThKhmM\ that $\Khm{\psi,\top,\chi}\in\Gamma$. It follows by Proposition \ref{prop:OneStepExistence} that there exists $\Delta''\in\Phi_\Gamma$ such that $\chi\in\Delta''$. Therefore, we have $t=(\Delta'',\chi^\psi)\in\calS^c$. Since there exists $\Delta'\in\Phi_\Gamma$ with $\phi\in\Delta'$, it follows by Proposition \ref{prop:allphiImplyUphi} that there is $t'\in\calS^c$ such that $\Lef{t'}=\Delta'$. Now, starting with any $\psi$-state, $a$ is clearly executable and it will lead to a $\chi$-state, and then by a $b$ step we will reach all the $\phi$ states. Therefore, by IH, we have that $ab$ is strongly $\chi$-executable at $w'$, and that for all $v'$ with $w'\rel{ab}v'$ we have $\M^c_\Gamma,v'\vDash\phi$. Therefore, we have $\M^c_\Gamma,w\vDash\Khm{\psi,\chi,\phi}$. Note that we do need a 2-step $\sigma$ in this case. 
		\end{itemize}
	\end{proof}
Now due to a standard Lindenbaum-like argument, each $\SKhm$-consistent set of formulas can be extended to a maximal consistent set $\Gamma$. Due to the truth lemma, $\M^c_\Gamma, (\Gamma, \top^\top)\vDash\Gamma.$ The completeness of $\SKhm$ follows immediately. 
\begin{theorem}
	$\SKhm$ is strongly complete w.r.t. the class of all models. 
\end{theorem}	

\section{Conclusions}
This paper generalizes the knowing how logic presented in \cite{Wang15lori} and 
proposes a ternary modal operator $\Khm{\psi,\chi,\phi}$ to express that the agent knows how to achieve $\phi$ given $\psi$ while maintaining $\chi$ in-between. This paper also presents a sound and complete axiomatization of this logic. Compared to the completeness proof in \cite{Wang15lori}, the proof here is much more complicated, and the  essential difference is that the state of the canonical model here is a pair consisting of a maximal consistent set and a marker of the form $\chi^\psi$ which indicates that this state has a $\lr{\psi,\bot,\chi}$-predecessor, in order to handle the intermediate constraints.

For future research, besides the obvious questions of decidability and model theory of the logic, we may give some alternative semantics to the same language by relaxing the strong executability. Intuitively, strongly executable plan may be too strong for knowledge-how in some cases. For example, if there is an action sequence $\sigma$ in the agent's ability map such that doing $\sigma$ at a $\psi$-state will always make the agent \textit{stop} on $\phi$ states, we can probably also say the agent knows how to achieve $\phi$ given $\psi$, e.g., I know how to start the engine in that old car, just turn the key several times until it starts,  and three times should suffice at most. Please note that there are two kinds of states on which the agent might stop: either states the agent achieves after doing $\sigma$ successfully, or states on which the agent is unable to continue executing the remaining actions. 

Another interesting topic is extending this logic with public announcement operators. Intuitively, $[\theta]\phi$ says that $\phi$ holds after the information $\theta$ is provided. The update of the new information amounts to the change of the background knowledge throughout the model, and this will affect the knowledge-how. For example, a doctor may not know how to treat a patient with the disease $p$ since he is worried that the only available medicine may potentially cause some very bad side-effect $r$, which can be expressed as $\neg\Khm{p,\neg r,\neg p}$. Suppose a new scientific discovery shows that the side-effect is not possible under the relevant circumstance, then the doctor should know how to treat the patient, which can be expresses as $[\neg r]\Khm{p,\neg r,\neg p}$.\footnote{However, the announcement operator $[\phi]$ is not reducible in $\LKhm$ as discussed in the full version of \cite{Wang15lori} which is under submission.}

Moreover, we can consider contingent plans which involve conditions based on the knowledge of the agent. A contingent plan is a partial function on the agent's belief space. Such plans make more sense when the agent has the ability of observations during the execution of the plan. To consider contingent plan, we need to extend the model (ability map) with an epistemic relation. We then can express knowledge-that and knowledge-how at the same time, and discuss their interactions in one unified logical framework. 

\bibliography{kh}

\begin{thebibliography}{10}

\bibitem{KandA15}
Thomas {\AA}gotnes, Valentin Goranko, Wojciech Jamroga, and Michael Wooldridge.
\newblock Knowledge and ability.
\newblock In Hans van Ditmarsch, Joseph Halpern, Wiebe van~der Hoek, and
  Barteld Kooi, editors, {\em Handbook of Epistemic Logic}, chapter~11, pages
  543--589. College Publications, 2015.

\bibitem{FWvD14}
Jie Fan, Yanjing Wang, and Hans van Ditmarsch.
\newblock Almost necessary.
\newblock In {\em Advances in Modal Logic Vol.10}, pages 178--196, 2014.

\bibitem{FWvD15}
Jie Fan, Yanjing Wang, and Hans van Ditmarsch.
\newblock Contingency and knowing whether.
\newblock {\em The Review of Symbolic Logic}, 8:75--107, 2015.

\bibitem{Gochet13}
Paul Gochet.
\newblock An open problem in the logic of knowing how.
\newblock In J.~Hintikka, editor, {\em Open Problems in Epistemology}. The
  Philosophical Society of Finland, 2013.

\bibitem{GW16}
Tao Gu and Yanjing Wang.
\newblock ``knowing value'' logic as a normal modal logic.
\newblock In {\em Advances in Modal Logic Vol.11}, 2016.
\newblock forthcoming.

\bibitem{Hintikka:kab}
J.~Hintikka.
\newblock {\em Knowledge and Belief: An Introduction to the Logic of the Two
  Notions}.
\newblock Cornell University Press, Ithaca N.Y., 1962.

\bibitem{LauWang}
Tszyuen Lau and Yanjing Wang.
\newblock Knowing your ability.
\newblock {\em The Philosophical Forum}, 2016.
\newblock forthcoming.

\bibitem{SW98}
David~E. Smith and Daniel~S. Weld.
\newblock Conformant graphplan.
\newblock In {\em AAAI 98}, pages 889--896, 1998.

\bibitem{Wright51}
G.~H. Von~Wright.
\newblock {\em An Essay in Modal Logic}.
\newblock North Holland, Amsterdam, 1951.

\bibitem{Wang15lori}
Yanjing Wang.
\newblock A logic of knowing how.
\newblock In {\em Proceedings of LORI 2015}, pages 392--405, 2015.

\bibitem{Wang15}
Yanjing Wang.
\newblock Representing imperfect information of procedures with hyper models.
\newblock In {\em Proceedings of ICLA 2015}, pages 218--231, 2015.

\bibitem{Wang16}
Yanjing Wang.
\newblock Beyond knowing that: a new generation of epistemic logics.
\newblock In Hans van Ditmarsch and Gabriel Sandu, editors, {\em Jaakko
  Hintikka on knowledge and game theoretical semantics}. Springer, 2016.
\newblock forthcoming.

\bibitem{WangF13}
Yanjing Wang and Jie Fan.
\newblock Knowing that, knowing what, and public communication: Public
  announcement logic with {Kv} operators.
\newblock In {\em Proceedings of IJCAI 13}, pages 1147--1154, 2013.

\bibitem{WangF14}
Yanjing Wang and Jie Fan.
\newblock Conditionally knowing what.
\newblock In {\em Advances in Modal Logic Vol. 10}, pages 569--587, 2014.

\bibitem{YLW15}
Quan Yu, Yanjun Li, and Yanjing Wang.
\newblock A dynamic epistemic framework for conformant planning.
\newblock In {\em Proceedings of TARK'15}, pages 298--318. EPTCS, 2016.

\end{thebibliography}
\bibliographystyle{plain}

\end{document}